\documentclass{article}

\usepackage{arxiv}
\usepackage[utf8]{inputenc} 
\usepackage[T1]{fontenc}    
\usepackage{hyperref}       
\usepackage{booktabs}       
\usepackage{amsfonts}       
\usepackage{nicefrac}       
\usepackage{microtype}      
\usepackage{graphicx}
\usepackage{orcidlink}

\usepackage{amsmath,amssymb,amsthm,mathtools,amsfonts}
\usepackage{algorithm}
\usepackage[noend]{algpseudocode}
\usepackage{enumitem}
\usepackage{cite}
\usepackage{url}
\usepackage{stackengine}

\newtheorem{definition}{Definition}[section]
\newtheorem{theorem}[definition]{Theorem}
\newtheorem{lemma}[definition]{Lemma}

\newtheorem{proposition}[definition]{Proposition}
\theoremstyle{remark}
\newtheorem{remark}[definition]{Remark}
\newtheorem{example}[definition]{Example}

\newcommand{\U}{\mathcal{U}}

\newcommand{\Rank}{\mathsf{Rank}}
\newcommand{\Select}{\mathsf{Select}}
\newcommand{\Enumerate}{\mathsf{Enumerate}}
\newcommand{\Insert}{\mathsf{Insert}}
\newcommand{\Delete}{\mathsf{Delete}}
\newcommand{\Aggregate}{\mathsf{Aggregate}}
\newcommand{\range}[2]{[#1,#2)}
\newcommand{\fp}{\mathsf{fp}}
\newcommand{\fpne}{\mathsf{fp}_\mathrm{NE}}

\newcommand{\Diff}{\Delta}

\newcommand{\weightsum}{\otimes}

\usepackage[dvipsnames]{xcolor}

\title{
    Range-Based Set Reconciliation
    via Range-Summarizable Order-Statistics Stores
}
\author{
    Elvio G. Amparore\orcidlink{0000-0003-1147-8985}\\
    Dipartimento di Informatica\\
    Università di Torino\\
    C.so Svizzera 185,\; 10149,\; Torino,\; Italy\\
    \url{elviogilberto.amparore@unito.it}
}
\date{March 2026}

\begin{document}
\maketitle

\begin{abstract}
Range-Based Set Reconciliation (RBSR) synchronizes ordered sets by
recursively comparing summaries of contiguous ranges and refining only the
mismatching parts. While its communication complexity is well understood,
its local computational cost fundamentally depends on the storage backend that
must answer repeated range-summary, rank, and enumeration queries during
refinement.

We argue that a natural storage abstraction for RBSR implementations
based on composable range aggregates is a
\emph{range-summarizable order-statistics store} (RSOS): a dynamic
ordered-set structure supporting composable summaries of contiguous
ranges together with rank/select navigation.
This identifies and formalizes the
backend contract needed for efficient recursive refinement, combining
range-summary support with order-statistics navigation for balanced
partitioning. We then show that a specific augmentation of
B\textsuperscript{+}-trees with subtree counts and composable summaries
realizes a RSOS, and we derive corresponding bounds on local reconciliation
work in this abstract storage model.

Finally, we introduce AELMDB, an extension of LMDB that realizes this
design inside a persistent memory-mapped engine, and evaluate it through
an integration with Negentropy.
The results show that placing the
reconciliation oracle inside the storage tree substantially reduces local
reconciliation cost on the evaluated reconciliation-heavy workloads
compared with an open-source persistent baseline based on auxiliary tree
caches, while the window-subrange ablation further confirms the usefulness
of the systems optimizations built on top of the core aggregate
representation.
\end{abstract}

\keywords{
    Range-based set reconciliation
    \and
    Anti-entropy synchronization
    \and
    Order-statistics
    \and
    B\textsuperscript{+}-tree aggregates
}

\section{Introduction}
\label{sec:introduction}

Set reconciliation asks two parties holding finite sets to determine which
elements are missing on either side. The problem appears throughout
replication, anti-entropy repair, peer-to-peer dissemination, and state
transfer. In these settings the total state may be large while the symmetric
difference between replicas is often small, so retransmitting the full state
is wasteful. The usual goal is therefore to make communication depend mainly
on the difference while keeping local computation manageable.

Classical approaches pursue near-optimal communication for sparse
differences, often in one or a few rounds. Characteristic-polynomial
schemes, difference digests, invertible Bloom lookup tables, and more recent
rateless variants belong to this line of work
\cite{minsky2003set,eppstein2011whats,goodrich2011iblt,yang2024rateless}.
These methods are powerful, but their best operating point often requires
substantial sketching or decoding work.

Range-Based Set Reconciliation (RBSR) takes a different path. Instead of
solving the whole sparse-difference problem through a single global sketch,
it recursively partitions a totally ordered universe into contiguous ranges,
compares protocol-level range values derived from those ranges, and
descends only where the compared values indicate a mismatch
\cite{meyer2023rbsr}. This yields logarithmically many rounds and
communication within a logarithmic factor of optimal. However, the
practical efficiency of RBSR depends critically on the storage backend:
the protocol is attractive only if the backend can summarize an arbitrary
range, split that range by relative cardinality, and enumerate small
residual parts without repeated scanning.

Practical RBSR, through the Negentropy protocol\cite{negentropy}, is already relevant 
and in use, in particular in Web3-related systems.
It has been adopted for Nostr synchronization through NIP-77\cite{nip77} and 
implemented in software such as \texttt{strfry}\cite{strfry}, 
\texttt{nostria-relay}\cite{nostriarelay}, and the
\texttt{@nostr-dev-kit/sync} package\cite{ndksync}.
It has also been discussed as a basis for Waku Sync\cite{waku_sync,waku_sync_forum}.
These developments suggest that range-based reconciliation is attracting interest as a
practical synchronization primitive in ordered, append-heavy distributed
systems.


This paper isolates that backend requirement as an explicit storage
abstraction. We argue that the appropriate interface for RBSR is a
range-summarizable order-statistics store (RSOS): a dynamic ordered-set
structure supporting (i) summaries of arbitrary contiguous ranges and
(ii) navigation by rank. The first component supplies the composable
aggregate information needed by recursive reconciliation, while the second
supports the balanced recursive partitioning used by the protocol. In this
view, efficient range-based reconciliation depends on a backend contract
that combines range summarization with order-statistics support.

Once these backend requirements are formulated as RSOS, a natural
realization is to augment a page-oriented B\textsuperscript{+}-tree with
subtree counts and composable range summaries. The leaves already store the
ordered set; the internal nodes can cache the metadata needed for
reconciliation; and the page-oriented structure matches external-memory and
memory-mapped storage. This yields a coherent design in which the same tree
stores both the data and the auxiliary metadata required for reconciliation.

The paper makes four contributions.
\begin{enumerate}[leftmargin=*,itemsep=2pt]
    \item It isolates and formalizes the RSOS abstraction, identifying the
    backend operations required to support RBSR efficiently.

    \item It shows that aggregate-augmented B\textsuperscript{+}-trees realize
    RSOS through standard subtree-count and composable-summary augmentation,
    and derives the resulting local-work bounds for reconciliation over such
    stores.
    
    \item It presents AELMDB, an extension of the open-source database LMDB
    \cite{lmdb}, as a concrete realization of this RSOS design.

    \item It evaluates AELMDB through an integration with the RBSR protocol
    implementation Negentropy, showing that the proposed storage design yields
    substantial reconciliation-time gains on the evaluated workload families.
\end{enumerate}

The rest of the paper proceeds from theory to systems. Section
\ref{sec:related} situates the work. Sections \ref{sec:model}--\ref{sec:bptree}
develop the formal model, protocol abstraction, and B\textsuperscript{+}-tree
realization. Section \ref{sec:systems} connects the theory to Negentropy and
AELMDB. Section \ref{sec:evaluation} summarizes the evaluation, and Sections
\ref{sec:discussion} and \ref{sec:conclusion} close the paper.

\section{Related Work}
\label{sec:related}

This section positions the paper along four nearby lines of work: classical
set reconciliation by sketches, range-based reconciliation and its requirements, 
authenticated tree-based structures, and practical protocol
deployments built around Negentropy.

\subsection{Classical set reconciliation}

A major line of research in set reconciliation aims to make communication
depend primarily on the size of the symmetric difference.
A central family of approaches is based on \emph{sketches}: rather than
exploiting ordered structure directly, each party computes a compact
summary from which the difference can later be recovered, often in one
round or through a coded stream of exchanged symbols.
Characteristic-polynomial interpolation achieves nearly optimal
communication for sparse differences \cite{minsky2003set}.
Subsequent work proposed more practical sketch-based techniques,
including difference digests and invertible Bloom lookup tables
\cite{eppstein2011whats,goodrich2011iblt}.
More recently, rateless encodings have been introduced to remain effective
across a wider range of difference sizes without requiring an accurate
estimate in advance \cite{yang2024rateless}.

Two more recent comparison points are worth noting explicitly.
PBS targets a different communication/computation trade-off, aiming for
low computational cost together with communication close to the
information-theoretic minimum \cite{gong2021pbs}.
CertainSync pushes in yet another direction, proposing rateless set
reconciliation with certainty and without parameter estimation
\cite{keniagin2025certainsync}.
These schemes remain sketch-centric: they treat reconciliation primarily
as a global coding problem over the set difference.

RBSR follows a different strategy.
Rather than encoding the whole sparse-difference problem into a single
global sketch, it recursively refines an ordered domain by comparing
protocol-level range values derived from contiguous ranges and descending
only on mismatching subranges.
This typically requires logarithmically more interaction, but it also
changes the computational profile of reconciliation: efficiency now
depends on the ability to summarize arbitrary ordered ranges, split them
by relative cardinality, and enumerate only the small residual
mismatches \cite{meyer2023rbsr}.
RBSR is especially useful in continuous synchronization settings, and more
generally when differences are concentrated in a small number of
contiguous regions, because only a limited part of the recursion tree
then needs to be explored.

\subsection{Range-based reconciliation and backend requirements}

The most systematic generic treatment of the RBSR family is due to
Meyer~\cite{meyer2023rbsr}.
Three features of that framework are especially important for the present
paper: reconciliation is defined over a totally ordered universe,
mismatching regions are refined through balanced range partitioning, and
efficient local execution depends on support for both composable range
summaries and subtree cardinalities.
We take this framework as the starting point for our formalization.

Earlier divide-and-conquer and partition-oriented ideas also appear in the
older scalable/practical set-reconciliation line of Minsky and
Trachtenberg \cite{minsky2002scalable}.
Our focus differs from those works in that we isolate the storage
interface required by recursive ordered refinement and study its
realization inside a persistent ordered index.

A complementary line of work shows that RBSR need not rely on one
specific summary-comparison mechanism.
Meyer and Scherer demonstrate that RBSR can also be realized using
conventional cryptographic hashes over history-independent,
clamping-invariant search trees \cite{meyer2024nonhomomorphic}.
This broadens the protocol-level design space.
Our focus is orthogonal: rather than asking which comparison constructions
are possible in general, we ask which dynamic storage abstraction most
naturally matches the operational requirements of range-based
reconciliation and how that abstraction can be realized efficiently in
persistent storage.

At the data-structure level, the mechanisms used here are related to
standard order-statistics augmentations of balanced search trees and to
counted or aggregate-carrying variants of page-oriented search trees
\cite{cormen2022intro,tatham2004countedbtrees,roura2001newmethod}.
Recent work on AB-tree provides evidence that that maintaining
aggregate metadata inside page-oriented search trees
supports efficient sampling under updates \cite{zhao2022abtree}.
Our contribution is therefore not the introduction of subtree-count
augmentation in isolation, but the observation that efficient RBSR
requires a precise bundle of capabilities: composable range summaries,
rank/select navigation, and small-range enumeration.
These can be realized coherently inside one persistent ordered index.

\subsection{Authenticated trees and canonical tree representations}

Merkle trees and authenticated search structures are natural comparison
points, since they also cache subtree digests and support recursive
comparison.
Merkle Search Trees extend this idea to ordered CRDT states in open
networks, but they do so by imposing a rigid, canonically derived tree
representation of the data \cite{auvolat2019mst}.
This is a useful comparison precisely because it highlights a key
difference with RBSR.

Authenticated B-tree literature is also relevant here.
In particular, Li et al.\ introduce dynamic authenticated page-oriented
indexes, including the Embedded Merkle B-tree, for verifiable query
answering over outsourced databases \cite{li2006dynamicauth}.
These structures share with our setting the idea of placing cached
summary/authentication metadata inside a B-tree-family index, but their
goal is different: they support authenticated query proofs, whereas RBSR
needs efficient range aggregation and balanced cuts by item rank.

RBSR does not reconcile two fixed tree shapes.
Instead, it operates on \emph{logical ranges of an ordered set}.
The recursion is therefore driven by the order of the elements
themselves, not by the topology of a canonical tree representation.
Balanced refinement is expressed in terms of cardinality within a range,
which makes order statistics central to the storage interface.
This distinction is important for our storage-theoretic view: the
question is not how to authenticate a predetermined tree, but how to
maintain a dynamic ordered set that supports efficient range summaries
and balanced cuts.

\subsection{Practical protocols and deployed systems}

Negentropy \cite{negentropy,nip77} is a practical instantiation of RBSR. It
fixes a concrete ordered record model together with 256-bit identifiers
that contribute to the composable summary carried by contiguous ranges.

Negentropy is already relevant at the systems level. In Nostr, NIP-77
specifies a Negentropy-based synchronization wrapper for both
client$\leftrightarrow$relay and relay$\leftrightarrow$relay exchange, and
practical implementations exist in software such as the \texttt{strfry}
relay, the \texttt{nostria-relay} project, and the
\texttt{@nostr-dev-kit/sync} package
\cite{negentropy,nip77,strfry,nostriarelay,ndksync}. It has also been
explored as a basis for Waku Sync
\cite{waku_sync,waku_sync_forum}.
This emerging adoption in Web3-related
systems suggests that range-based set reconciliation may become a 
common foundational synchronization primitive in future protocol stacks.

Negentropy is therefore the natural protocol-level reference for our study:
it is practical, it instantiates the abstract model developed in this paper,
and improvements to its storage layer have direct systems relevance.
Existing LMDB-based implementations realize its required range summaries and
order-statistics queries through an auxiliary tree with cached aggregates;
this serves as the baseline for comparison with the in-tree aggregate design
studied here.

\section{Model and Abstract Oracle}
\label{sec:model}

We formalize RBSR through an abstract storage oracle. The goal of this
section is to isolate the exact local capabilities that range-based
reconciliation requires: range aggregation, order statistics, and explicit
enumeration on small residual ranges.

\subsection{Ordered reconciliation instances}

We begin with the ordered-set model underlying range-based reconciliation.

\begin{definition}[Ordered universe]
Let $(\U,\preceq)$ be a totally ordered universe, and let $\prec$ denote the
associated strict order.
\end{definition}

\begin{definition}[Replica state]
A replica state is a finite subset $X \subseteq \U$. 
For admissible bounds $l,u$ delimiting a contiguous slice of \(\U\), define
\[
X \cap \range{l}{u} := \{x \in X \mid l \preceq x \prec u\}.
\]
\end{definition}
To keep the interval notation flexible at the application level, we allow
the bound values \(l\) and \(u\) to be either elements of \(U\) or external
sentinel/cut values, provided that \(\range{l}{u}\) denotes a contiguous
subset of \(U\).
Typical examples are application-defined outer endpoints, timestamp cut
values, or context-dependent sentinels such as \(-\infty\) and \(+\infty\).

\begin{definition}[Reconciliation instance]
An ordered reconciliation instance is a pair $(X,Y)$ of finite subsets of
$\U$. Its global symmetric difference is
\[
\Diff(X,Y) := (X \setminus Y) \cup (Y \setminus X),
\]
and its range-local symmetric difference is
\[
\Diff_{l,u}(X,Y)
:= \Diff\bigl(X \cap \range{l}{u},\; Y \cap \range{l}{u}\bigr).
\]
\end{definition}

The \emph{reconciliation task} is to compute the global symmetric difference
$\Diff(X,Y)$, or, when restricted to a half-open range $\range{l}{u}$,
the local symmetric difference $\Diff_{l,u}(X,Y)$.
From the perspective of replica $X$, the set of items it \emph{has} but
$Y$ lacks is $X \setminus Y$, while the set of items it \emph{needs} from
$Y$ is $Y \setminus X$.

\subsection{Composable range aggregates and comparison maps}

RBSR compares contiguous ranges by means of compact composable summaries. We
first define the algebra carried by the elements, then lift it to whole
ranges.

\begin{definition}[Element-summary monoid]
\label{def:summary-monoid}
Let $(M,\oplus,0_M)$ be a monoid and let $\phi : \U \to M$. For a finite
ordered subset $S = \{x_1 \prec \cdots \prec x_k\}$, define its summary
\[
\Sigma(S) := \phi(x_1) \oplus \cdots \oplus \phi(x_k),
\qquad
\Sigma(\varnothing) := 0_M.
\]
\end{definition}

Because $\oplus$ is associative, the summary of a range can be obtained by
combining the summaries of subranges. This is the algebraic property that
makes caching possible.

RBSR, however, needs not only a composable summary value but also the
cardinality of the range being summarized, since balanced refinement is
defined by relative rank. We therefore bundle both quantities into a single
aggregate.

\begin{definition}[Range aggregate]\label{def:rangeaggregate}
Given an element-summary monoid $(M,\oplus,0_M)$, define the aggregate
monoid
\[
\mathcal{A} := \bigl(\mathbb{N}\times M,\; \weightsum,\; (0,0_M)\bigr),
\]
with
\[
(c_1,m_1)\weightsum(c_2,m_2)
:= (c_1+c_2,\; m_1 \oplus m_2).
\]
For a finite set $S$, its aggregate is
\[
A(S) := \bigl(|S|,\Sigma(S)\bigr) \in \mathcal{A}.
\]
\end{definition}

The aggregate $A(S)$ is the fundamental range object in the model: it is
composable, cacheable, and already contains the two quantities that RBSR
needs locally.

\begin{definition}[Range comparison map]
Let \(\mathcal{A}\) be the aggregate space from Definition~\ref{def:rangeaggregate}.
A \emph{range comparison map} is a protocol-specific function
\[
\fp : \mathcal{A} \to F
\]
into some comparison space \(F\).

For a queried range \(\range{l}{u}\), the protocol compares
\[
\fp\!\bigl(A(X \cap \range{l}{u})\bigr)
\quad\text{and}\quad
\fp\!\bigl(A(Y \cap \range{l}{u})\bigr).
\]
If these values agree, the protocol may treat the queried range as matched;
otherwise it continues with refinement or explicit enumeration.
\end{definition}

This separates two roles in the protocol.
The aggregate \(A(\cdot)\) is the composable storage-level object maintained
by the backend, combining the range cardinality and the summary value needed
locally by RBSR.
The comparison map \(\fp\) is instead part of the protocol layer: it
turns the aggregates into the values that are actually compared during reconciliation.

The purpose of the present paper is to identify and realize the storage-side
aggregate object \(A(\cdot)\) and the operations needed to maintain and query it efficiently.
The specific design or analysis of the protocol-level comparison map is not
part of the RSOS abstraction developed here; it is taken "as-is" from the concrete
RBSR instantiation under study.

An illustrative example of the aggregate object is presented in
Appendix~\ref{sec:example:rangeaggr}.

\subsection{Order statistics}

The second requirement of RBSR is the ability to cut a range by
\emph{cardinality}, not only by key value. This is why order statistics are
part of the required interface.

\begin{definition}[Rank and select]
Let $X = \{x_0 \prec x_1 \prec \cdots \prec x_{n-1}\}$. 
For a bound value $z$, define
\[
\Rank_X(z) := |\{x \in X \mid x \prec z\}|.
\]
For $r \in \{0,\dots,n-1\}$, define
\[
\Select_X(r) := x_r.
\]
\end{definition}

With these order primitives, we can define how to proceed recursively in partitioning ranges.

\begin{definition}[Balanced $b$-partition]
Let $m := |X \cap \range{l}{u}|$. A balanced $b$-partition of
$\range{l}{u}$ in $X$ is a sequence of boundary values
\[
l = c_0 \preceq c_1 \preceq \cdots \preceq c_b = u
\]
such that the induced subranges
\[
\range{c_0}{c_1},\; \range{c_1}{c_2},\; \dots,\; \range{c_{b-1}}{c_b}
\]
form $b$ consecutive parts, each containing either
$\lfloor m/b \rfloor$ or $\lceil m/b \rceil$ elements of
$X \cap \range{l}{u}$.

Equivalently, for each $j=0,\dots,b$, the cut $c_j$ may be chosen so that
exactly $\left\lfloor \frac{j\,m}{b} \right\rfloor$ elements of
$X \cap \range{l}{u}$ are strictly smaller than $c_j$, with consecutive
boundaries allowed to coincide when needed.
\end{definition}

In implementations, coincident consecutive boundaries may be normalized/removed away, 
so the actual recursive call set is the family of \emph{nonempty} subranges 
induced by the balanced partition.
Balanced recursive refinement is naturally implemented through rank and
select: rank gives positions inside a range, and select materializes the
corresponding cut values.

\subsection{Range-summarizable order-statistics stores}

We can now state the abstract storage interface required by RBSR.

\begin{definition}[RSOS]
A \emph{Range-Summarizable Order-statistics Store} over $\U$ is a
dynamic data structure representing a finite set $X \subseteq \U$ and
supporting
\begin{align*}
\mathsf{size}() &\to |X|,\\
\Aggregate(l,u) &\to A(X \cap \range{l}{u}),\\
\Rank(z) &\to \Rank_X(z),\\
\Select(r) &\to \Select_X(r),\\
\Enumerate(l,u) &\to \text{the ordered contents of } X \cap \range{l}{u},
\end{align*}
together with $\Insert(x)$ and $\Delete(x)$ to update $X$.
\end{definition}

The output of $\Aggregate(l,u)$ already combines the two local quantities
needed by RBSR: cardinality, used for balanced partitioning, and composable
summary, used to compare corresponding ranges. The remaining operations
provide the navigation and fallback enumeration needed by the protocol.

\section{RBSR over RSOS}
\label{sec:protocol}

We now formulate RBSR over the abstract RSOS interface, independently of any
particular storage representation.

Let $O_X$ and $O_Y$ be RSOS oracles representing two replica states
$X,Y \subseteq \U$ over the same ordered universe. The protocol maintains a
finite family of pairwise disjoint \emph{active ranges}. Initially, this
family consists of one application-defined outer range covering the part of
the universe to be reconciled. During execution, each active range is either
resolved, explicitly enumerated, or replaced by a balanced family of child
ranges. Thus the active ranges always form a partition of the still
unresolved portion of the reconciliation instance.

We describe one responder step from the point of view of peer \(X\),
replying to a query sent by peer \(Y\). For each queried range
\(\range{l}{u}\), peer \(Y\) provides its local comparison value
\[
f_Y = \fp\bigl(O_Y.\Aggregate(l,u)\bigr).
\]
Peer \(X\) computes the local aggregate of \(X \cap \range{l}{u}\) and then
does one of three things:
\begin{enumerate}
  \item if the comparison values match, it returns \textsc{Skip};
  \item if they do not match and \( |X \cap \range{l}{u}| \le t \), it
  returns \textsc{IdList} with the explicit ordered contents of the range;
  \item otherwise, it returns the comparison values of a \(b\)-balanced
  partition of \(\range{l}{u}\).
\end{enumerate}

\begin{algorithm}[t]
\caption{Responder step at peer $X$}
\label{alg:rbsr-step}
\begin{algorithmic}[1]
\Procedure{Respond}{$O_X,l,u,f_Y,b,t$}
    \State $a_X \gets O_X.\Aggregate(l,u)$
    \State $(m_X,\sigma_X) \gets a_X$
    \State $f_X \gets \fp(a_X)$
    \If{$f_X = f_Y$}
        \State \Return \textsc{Skip}$(\range{l}{u})$
    \EndIf
    \If{$m_X \le t$}
        \State \Return \textsc{IdList}$(\range{l}{u}, O_X.\Enumerate(l,u))$
    \EndIf
    \State $C \gets \Call{SplitByRank}{O_X,l,u,b}$
    \State \Return $\textsc{Split}\bigl([\bigl(\range{C_j}{C_{j+1}},\, \fp(O_X.\Aggregate(C_j,C_{j+1}))\bigr)]_j\bigr)$
\EndProcedure
\end{algorithmic}
\end{algorithm}

Algorithm \ref{alg:rbsr-step} sketches the response procedure.
The partitioning procedure, listed in Algorithm~\ref{alg:split-by-rank},
uses the count component of the aggregate to choose
cuts by relative rank inside the queried range.
Function $\Call{NormalizeCuts}{C}$ removes consecutive duplicate cuts, so that
zero-width child ranges are dropped. The remaining child ranges are still
pairwise disjoint and their union is the parent range $\range{l}{u}$.

\begin{algorithm}[t]
\caption{Balanced splitting by rank}
\label{alg:split-by-rank}
\begin{algorithmic}[1]
\Procedure{SplitByRank}{$O,l,u,b$}
    \State $(m,\_) \gets O.\Aggregate(l,u)$
    \If{$m=0$}
        \State \Return $[l,u]$
    \EndIf
    \State $r_0 \gets O.\Rank(l)$
    \State $C \gets [l]$
    \For{$j=1$ to $b-1$}
        \State $q_j \gets \left\lfloor \tfrac{j\,m}{b} \right\rfloor$
        \State append $O.\Select(r_0 + q_j)$ to $C$
    \EndFor
    \State append $u$ to $C$
    \State \Return \Call{NormalizeCuts}{$C$}
\EndProcedure
\end{algorithmic}
\end{algorithm}

A complete protocol round consists of each peer applying
Algorithm~\ref{alg:rbsr-step} to the active ranges it is asked to answer.
A range is resolved immediately in the \textsc{Skip} case.
In the \textsc{IdList} case, the range is resolved exactly once the receiving
peer compares the received ordered list with its own local ordered contents of
the same range.
This exact comparison may itself require a local call to \(\Enumerate(l,u)\) on the
receiving side.
In the \textsc{Split} case, the parent range is removed from the active family
and replaced by the returned child ranges.

For fixed parameters $b \ge 2$ and $t \ge 1$, repeated refinement is finite:
every nonterminal mismatching range is replaced by proper child subranges,
and all cut values are drawn from a finite set consisting of endpoint values
and keys already present in one of the replicas. Hence the protocol
eventually reaches only \textsc{Skip} or \textsc{IdList} leaves, and then terminates.

\begin{proposition}[Protocol correctness under sound skip decisions]
\label{prop:correctness}
Assume that whenever the protocol returns \textsc{Skip} on a queried range
\(\range{l}{u}\), one has
\[
X \cap \range{l}{u} = Y \cap \range{l}{u}.
\]
Then RBSR computes the exact symmetric difference \(\Delta(X,Y)\).
\end{proposition}

\begin{proof}
Consider any active range \(\range{l}{u}\).

If the protocol returns \textsc{Skip}, then by hypothesis
\[
X \cap \range{l}{u} = Y \cap \range{l}{u},
\]
so the local symmetric difference on that range is empty.

If the protocol returns \textsc{IdList}, then one peer explicitly sends the
ordered contents of its local subset on \(\range{l}{u}\). The receiving
peer compares that list with its own ordered contents on the same range and
thereby recovers the exact local symmetric difference on \(\range{l}{u}\).

If the protocol returns \textsc{Split}, then the child ranges are pairwise
disjoint and their union is exactly the parent range. Therefore the local
symmetric difference on the parent range is the disjoint union of the local
symmetric differences on the child ranges. By induction over the finite
refinement tree, every leaf range is resolved exactly, and therefore their
union yields the exact global symmetric difference \(\Delta(X,Y)\).
\end{proof}

Proposition~\ref{prop:correctness} makes explicit the only protocol-level
assumption needed for exact reconciliation. Namely, whenever the protocol
stops recursion with \textsc{Skip}, that decision must be sound for the
queried range. 
The storage-side requirements are independent of how a concrete protocol
realizes that test: RBSR still needs a backend that can aggregate arbitrary
ordered ranges, map keys to ranks, select cut points by relative position,
and enumerate small residual ranges exactly.

\section{Augmented B\textsuperscript{+}-Trees as a Natural Realization}
\label{sec:bptree}

We now move from the abstract RSOS interface to a concrete data-structural
realization. Prior work already observes that composable monoidal summaries
fit naturally with higher-degree search trees
\cite{meyer2023rbsr}. This motivates page-oriented balanced
trees, and in particular B\textsuperscript{+}-trees, as a natural persistent
realization.

Let $X=\{x_0 \prec x_1 \prec \dots \prec x_{n-1}\}$ be the ordered set stored
in the leaves of a balanced B\textsuperscript{+}-tree.

\begin{definition}[Aggregate-augmented B\textsuperscript{+}-tree]
An aggregate-augmented B\textsuperscript{+}-tree is a
B\textsuperscript{+}-tree in which every child reference of an internal node
stores the cached aggregate
\[
A(S)=(|S|,\Sigma(S)),
\]
where $S \subseteq X$ is the set of descendant items in that child subtree.
\end{definition}

Because the product aggregate $A(S)=(|S|,\Sigma(S))$ forms a monoid,
aggregates over larger subtrees are obtained by combining child aggregates,
while insertions and deletions can maintain these cached values by
propagating local changes along the modified search path and through the
possible upward rebalancing cascade of the tree.

This annotation realizes the RSOS operations directly.

For $\Rank(z)$, one performs the usual search for the lower bound of $z$.
Along the root-to-leaf path, whenever the search descends through a child,
the counts of all preceding sibling subtrees are added to an accumulator.
The final sum is exactly the number of stored keys strictly smaller than
$z$.

For $\Select(r)$, one descends from the root while maintaining a residual
rank. At each internal node, let the children from left to right represent
subsets $S_0,\dots,S_{k-1}$ with cached counts $|S_0|,\dots,|S_{k-1}|$.
Choose the unique child index $i$ such that
\[
\sum_{j<i}|S_j| \;\le\; r \;<\; \sum_{j\le i}|S_j|,
\]
subtract $\sum_{j<i}|S_j|$ from the residual rank, and recurse into child
$i$. At the leaf, the residual rank identifies the desired key.

For $\Aggregate(l,u)$, one performs the standard two-boundary range walk.
Subtrees fully contained in $\range{l}{u}$ contribute their cached aggregate
directly, while only the two boundary regions require further descent. Thus
range aggregation is obtained by combining a small number of cached subtree
aggregates with the aggregates of the boundary fragments.

For $\Enumerate(l,u)$, one seeks the first leaf position at or after $l$ and
then scans forward in leaf order until reaching $u$. Since
B\textsuperscript{+}-tree leaves are maintained in key order, this returns
the contents of $X \cap \range{l}{u}$ in sorted order.

Thus a single aggregate-augmented B\textsuperscript{+}-tree realizes all the
operations required by RSOS: the same tree stores the ordered set, supports
rank/select navigation through subtree counts, and answers range aggregates
through cached composable metadata.

\begin{theorem}[RSOS realization theorem]
\label{thm:bptree-rsos}
Let \(X\) be stored in a balanced aggregate-augmented
B\textsuperscript{+}-tree of height \(h=\Theta(\log_B n)\), where
\(n=|X|\) and \(B\) is the maximum number of child references that fit in
one internal page.
Assume the standard bounded-page-size model, so \(B=O(1)\) with respect
to \(n\), and that basic operations (aggregate-field combination, copying, 
fixed-size summaries, ...) take \(O(1)\) time.
Then this tree realizes an RSOS, and the supported operations satisfy:
\begin{itemize}
    \item \(\Rank(z)\) and \(\Select(r)\) take \(O(h)\) time;
    \item \(\Aggregate(l,u)\) takes \(O(Bh)\) time, hence \(O(h)\) time
    under bounded page size;
    \item \(\Enumerate(l,u)\) takes \(O(h+k)\) time, where
    \(k = |X \cap \range{l}{u}|\);
    \item \(\Insert(x)\) and \(\Delete(x)\) affect at most one
    root-to-leaf search path together with at most one split, merge, or
    redistribution per level, so maintaining the cached aggregates costs
    \(O(Bh)\), hence \(O(h)\), time under bounded page size.
\end{itemize}
\end{theorem}

\begin{proof}
For \(\Rank(z)\), the search follows one root-to-leaf path.
At each internal node, the counts of the sibling subtrees strictly to the
left of the chosen child are added to an accumulator.
Since exactly one child is followed per level, the total cost is \(O(h)\).

For \(\Select(r)\), one again follows a single root-to-leaf path.
At each internal node, the cached subtree counts determine the unique child
whose cumulative count interval contains the residual rank.
This also costs \(O(h)\).

For \(\Aggregate(l,u)\), only the two boundary search paths need explicit
descent.
Every subtree strictly between those two paths is fully covered by the
query range and contributes its cached aggregate directly.
At each of the \(h\) levels, at most \(O(B)\) child entries are inspected,
so the total work is \(O(Bh)\), which reduces to \(O(h)\) in the
bounded-page-size model.

For \(\Enumerate(l,u)\), one root-to-leaf descent finds the first relevant
leaf position, after which the output is produced by sequential scanning in
leaf order.
This yields \(O(h+k)\) time.

For updates, insertion or deletion changes one root-to-leaf path and may
trigger the usual rebalancing operations as the modification propagates
upward.
Thus at most \(O(h)\) nodes are structurally affected.
At each such node, refreshing the cached aggregate information requires
recomputing only a bounded number of child-derived summary fields, i.e.,
\(O(B)\) work per node.
Hence the total update-maintenance cost is \(O(Bh)\), which is \(O(h)\)
under bounded page size.
\end{proof}

This theorem is the storage-theoretic core of the paper. If a system already
stores its ordered state in a balanced B\textsuperscript{+}-tree, then adding
subtree counts together with fixed-size composable range aggregates yields
precisely the RSOS structure required by RBSR, without introducing a separate
auxiliary index.

\subsection{Bounds for reconciliation over RSOS B\textsuperscript{+}-trees}

Once reconciliation is expressed through the RSOS abstraction,
the cost of a reconciliation can be analyzed directly in terms of
the underlying B\textsuperscript{+}-tree operations.
We can derive per-step and per-reconciliation bounds:
\begin{enumerate}
    \item \textbf{Per-step cost.}
    In an RSOS B\textsuperscript{+}-tree, each reconciliation step costs
    $O(h)$, except for explicit enumeration, which costs $O(h+k)$.
    \\
    (derivation in Lemma~\ref{lem:single-step-cost} in the Appendix).

    \item \textbf{Execution-sensitive responder-side total cost.}
    Suppose the reconciliation tree contains $Q$ queried ranges and
    returns a total of $K$ explicitly enumerated items. Then
    \[
    T_{\mathrm{loc}} = O(bQh + K),
    \]
    and, for fixed $b$ and $t$, this simplifies to
    \[
    T_{\mathrm{loc}} = O(Qh).
    \]
    (derivation in Theorem~\ref{thm:execution-sensitive-cost} in the Appendix).
\end{enumerate}

\section{Implementation and System Models}
\label{sec:systems}

We now map the abstract RSOS model to a concrete two-layer design.
The protocol layer fixes the ordered record model together with the concrete
range-comparison rule used during reconciliation; the storage layer provides
the RSOS operations required by that protocol.

\subsection{Protocol layer}
At the protocol layer, we consider Negentropy
\cite{negentropy,logperiodic2023rbsr}, a practical realization of RBSR.
In Negentropy, each record consists of a timestamp and a 256-bit identifier; 
records are ordered lexicographically by \((\text{timestamp},
\text{identifier})\); this gives a total order while keeping timestamps as
the primary coordinate along which reconciliation ranges are sliced.
Identifiers contribute to the composable range summary.
For identifiers
\(\mathrm{id}_1,\ldots,\mathrm{id}_k\), Negentropy defines
\[
  \Sigma(\mathrm{id}_1,\ldots,\mathrm{id}_k) = 
    \mathrm{id}_1 + \ldots + \mathrm{id}_k \quad (\mathrm{mod}~2^{256}).
\]
Since modular addition is associative, the
identifier contribution forms an element-summary monoid. 

We write \(\fpne\) for Negentropy's concrete comparison value over range
aggregates. Concretely, Negentropy encodes the range aggregate
\[
A(S \cap \range{l}{u})
=
\bigl(|S \cap \range{l}{u}|,\; \Sigma(S \cap \range{l}{u})\bigr),
\]
hashes that encoding with SHA-256, and uses the first \(128\) bits as the
value compared by the protocol.
Therefore, Negentropy's concrete comparison rule \(\fpne\) should be read as
probabilistically sound rather than information-theoretically exact.
A false \textsc{Skip} may arise whenever two unequal ranges yield the same
comparison value, either because the underlying aggregate \(A(\cdot)\) is not
injective as a set summary, or because distinct encodings of aggregates
collide after hashing and truncation to \(128\) bits.
Proposition~\ref{prop:correctness} therefore applies to the
abstract protocol under a sound-skip assumption, while concrete Negentropy
adopts a tradeoff between probabilistic soundness and efficiency.
A full end-to-end collision analysis of \(\fpne\) is outside the scope of the
present paper. Here we keep the protocol rule fixed and focus on the storage
operations required to compute its input aggregate efficiently.


\subsection{Storage layer}

At the storage layer, the question is how to efficiently realize the RSOS oracle. 
The main design alternatives are: a materialized array, a separate auxiliary
tree layered over storage, or an in-tree aggregate design in which the
ordered storage structure itself maintains the counts and composable
summaries required by reconciliation.


The Anti-Entropy Lightning Memory-mapped Database (AELMDB)
\cite{aelmdb} is a newly developed extension of the open-source engine
LMDB~\cite{lmdb} that realizes this third design, and it is a major
software contribution of this paper.
AELMDB is not merely a benchmark-specific backend: it is a reusable LMDB
extension that exposes aggregate-aware ordered-set functionality to
applications while preserving LMDB's overall page-oriented architecture.
The contribution here is therefore the realization of the RSOS view
inside a persistent engine: the required metadata are integrated into the
on-disk page format, maintained by the update path, and surfaced through
a concrete API for range aggregates, rank/select operations, and
window-relative subranges.

AELMDB augments LMDB's B\textsuperscript{+}-tree by storing aggregate
metadata directly in branch pages, so that subtree counts and
range-summary values are maintained inside the same tree that stores the
records.
Aggregate support is selected per-database, with flags such as
\texttt{MDB\_AGG\_ENTRIES} and \texttt{MDB\_AGG\_KEYS} for counting
semantics, and \texttt{MDB\_AGG\_HASHSUM} for a fixed-size additive
summary.

\medskip

\begin{figure}[t]
\centering
\footnotesize
\setlength{\fboxsep}{5pt}
\begin{tabular}{l}
\textbf{meta page 0/1} \\
\fbox{\parbox{0.9\linewidth}{
\texttt{MDB\_meta} \; contains \; \texttt{mm\_dbs[FREE], mm\_dbs[MAIN]} \\
\texttt{MDB\_db}: usual LMDB tree metadata $+$ aggregate totals/configuration \\
\hspace*{1em}\texttt{md\_entries}, \underline{\texttt{md\_keys}}, \underline{\texttt{md\_hash\_offset}},
\underline{\texttt{md\_hashsum}}, \texttt{md\_root}
}} \\[0.6em]

\textbf{branch page} \\
\fbox{\parbox{0.9\linewidth}{
page header with \texttt{P\_AGG\_*} bits \;+\; pointer array \;+\; branch nodes \\
branch node layout: \texttt{[child page\_num | \underline{aggregates} | separator key]} \\
\underline{\texttt{aggregates}} := \underline{\texttt{[entries?][keys?][hashsum?]}}
}} \\[0.6em]

\textbf{leaf page, optional overflow pages} \\
\fbox{\parbox{0.9\linewidth}{
usual LMDB key/value records (payload stored here)
}}
\end{tabular}
\caption{Simplified AELMDB on-disk page layout. Aggregate metadata is stored in the
per-database descriptor and in branch-node prefixes, while leaf pages keep the
records themselves. Extensions of the LMDB format are \underline{underlined}.}
\label{fig:aelmdb-layout}
\end{figure}

AELMDB keeps LMDB's overall file organization (meta pages, branch pages,
leaf pages, and overflow pages) but extends the internal
B\textsuperscript{+}-tree layout with aggregate state (Figure \ref{fig:aelmdb-layout}). 
At the database level,
the persistent \texttt{MDB\_db} descriptor carries not only the usual root and
page counters, but also the aggregate totals and configuration needed by the
engine (notably the total entry count, optional distinct-key count, the hash
slice offset, and the running hashsum). Inside the tree, only branch pages are
modified: their headers record which aggregate components are active, and each
branch node stores, immediately before its separator key, a compact aggregate
prefix describing the referenced child subtree. In simplified form, an
aggregate-enabled branch node is
\texttt{[child pgno | aggregates | separator key]}, where the aggregates form a sequence of optional fields \texttt{[entries?][keys?][hashsum?]}. 
Leaf pages still store the actual key/value records. 
Thus AELMDB realizes RSOS by embedding the reconciliation
metadata directly into the same persistent tree that stores the ordered set,
rather than by maintaining a separate auxiliary structure.

\subsection{RSOS realization}

In AELMDB, the function \(\phi : \U \to M\) from
Definition~\ref{def:summary-monoid} is realized by extracting from each
record a fixed-size byte slice at a designated offset from its key or
value. The summary component \(\Sigma\) of a range is then obtained by
modular addition of the extracted slices over all records in that range,
while the count component is maintained by the aggregate metadata itself.
Hence AELMDB realizes the full range aggregate
\[
A(S)=(|S|,\Sigma(S))
\]
required by RSOS.

The correspondence with the abstract RSOS operations is direct:
\begin{align*}
\mathsf{size}() &\rightarrow \mathtt{mdb\_agg\_totals},\\
\Aggregate(l,u) &\rightarrow \mathtt{mdb\_agg\_range},\\
\Rank(z) &\rightarrow \mathtt{mdb\_agg\_rank},\\
\Select(r) &\rightarrow \mathtt{mdb\_agg\_select},\\
\Enumerate(l,u) &\rightarrow \text{iteration from } \mathtt{mdb\_agg\_cursor\_seek\_rank},\\
\Insert(x) &\rightarrow \mathtt{mdb\_put},\quad
\Delete(x) \rightarrow \mathtt{mdb\_del}.
\end{align*}

When keys are formed by a \((\text{timestamp}, 256\text{-bit identifier})\)
pair and \(\phi\) extracts the identifier slice, 
the storage-level summary algebra matches Negentropy's record model exactly:
lexicographic \((\text{timestamp}, \text{identifier})\) order defines the record order,
while identifiers provide the additive summary component.
In this configuration, AELMDB is a concrete realization of the
aggregate-augmented B\textsuperscript{+}-tree of
Section~\ref{sec:bptree}, instantiated to the aggregate algebra used by
Negentropy.

\subsection{Window Subranges}

AELMDB also provides a \emph{window subrange} interface tailored to the
recursive structure of RBSR. When many successive queries fall within the
same outer key range, an \texttt{MDB\_agg\_window} stores once the mapping
from that outer range to its absolute entry-rank interval. Subsequent
operations can then be expressed in ranks relative to that window, rather
than recomputing the outer bounds and their absolute positions at every
step.

Many aggregate and lower-bound queries in reconciliation loops are issued under stable outer bounds, so re-deriving the same key-range-to-rank mapping each time would
be redundant work. By reusing the stored window interval, AELMDB can answer
nested subrange aggregates and split-position queries more directly, and
thus improve locality and reduce repeated navigation overhead. This is not
part of the minimal RSOS abstraction, but a systems-level optimization for
the common RBSR pattern of many subrange queries within the same outer
range.

\section{Evaluation}
\label{sec:evaluation}

The purpose of the evaluation is to isolate the effect of
\emph{storage realization}, not to compare different reconciliation
protocols. Across all experiments, the protocol-level task is held fixed:
the same Negentropy/RBSR reconciliation workload is executed over different
backends implementing the same abstract oracle interface. We compare four
realizations: the baseline \texttt{BTreeLMDB}, the new aggregate-enabled
\texttt{AELMDB}, an ablation
\texttt{NoWndAELMDB}, and an in-memory \texttt{Vector}
backend used as an additional reference \cite{aelmdb_bench}.

The baseline \texttt{BTreeLMDB} realizes the required reconciliation
operations through a separate auxiliary tree stored inside LMDB. As a
consequence, range aggregates and rank-based navigation are not obtained
directly from LMDB's own page metadata; they require repeated traversal of
this auxiliary LMDB-backed tree, with per-node accumulation and weaker
locality. In effect, the hot reconciliation path is executed through a
``B-tree inside a B-tree.'' By contrast, \texttt{AELMDB} realizes the
same abstract operations using aggregate metadata stored directly in LMDB
branch pages, together with window subrange for repeated refinement.
The backend \texttt{NoWndAELMDB} uses the same aggregate-enabled
storage representation as \texttt{AELMDB}, but without the window subrange helper path; it therefore serves as an ablation test
of the extended window subrange interface.

\begin{table}[t]
\caption{Scenario-family parameters used in the benchmark suite. 
Each family runs 8 instances, \(i=1,\ldots,8\).
The last two columns aggregate the
outside-of-slice population over both sides of the slice.}
\label{tab:scenario-families}
\centering
\small
\begin{tabular}{|l|c|c|c|c|}
\hline
Family & 
\(|X_{\mathrm{in}}\cap Y_{\mathrm{in}}|\) &
\(|X_{\mathrm{in}}\setminus Y_{\mathrm{in}}|\) &
\(|X_{\mathrm{out}}\cap Y_{\mathrm{out}}|\) &
\(|X_{\mathrm{out}}\setminus Y_{\mathrm{out}}|\) \\
\hline
\texttt{base\_dense\(_i\)} & \(64i\) & \(4i\) & \(1000i\) & \(200i\) \\
\texttt{base\_sparse\(_i\)} & \(128i\) & \(6i\) & \(2000i\) & \(400i\) \\
\texttt{scale\_dense\(_i\)} & \(256i^2\) & \(8i\) & \(4000i\) & \(800i\) \\
\texttt{scale\_sparse\(_i\)} & \(512i\) & \(16i\) & \(6000i\) & \(1200i\) \\
\texttt{stress\(_i\)} & \(1024i^2\) & \(64i\) & \(8000i\) & \(1600i\) \\
\texttt{stress\_dyn\(_i\)} & \(4096i^2\) & \(1024i^2\) & \(4000i^2\) & \(800i^2\) \\
\hline
\end{tabular}
\end{table}

\subsection{Methodology}

The benchmark suite consists of synthetic reconciliation instances
\((X,Y)\) parameterized by dataset size and disagreement pattern. Each
instance is defined by a slice interval \(\range{l}{u}\), by the number
of common and peer-local elements inside that interval, and by the amount of
common and peer-local context outside the interval. The goal is to vary both
the total size of the reconciliation instance and the structure of the
symmetric difference \(\Diff(X,Y)\), while keeping the protocol logic
unchanged.

For each scenario and backend, the benchmark records preparation time,
bench-time opening/build/reconciliation costs, protocol-level message and
byte counts, estimated used storage, and Linux RSS sampled before and after
the bench run. 
The most relevant metrics are:
\(T_{\mathrm{prep}}\),
\(T_{\mathrm{rec}}\),
\(S_{\mathrm{disk}}\), and
\(S_{\mathrm{RSS}}\).
Here \(T_{\mathrm{prep}}\) is the one-time cost of preparing the backend for
a fixed scenario; \(T_{\mathrm{rec}}\) is the cost of a full reconciliation
run including \(\Diff(X,Y)\) materialization (called \emph{have/need} sets); 
\(S_{\mathrm{disk}}\) is the backend-local disk size used; 
and \(S_{\mathrm{RSS}}\) is
the resident set size sampled after the isolated bench run. 
We also report an in-memory non-persistent \texttt{Vector} backend for comparison, however the comparison is focused on persistent backends.

The benchmark was run on an AMD 3700X with Linux, using file \texttt{mmap} size of
\(2048\,\mathrm{MiB}\), and \(10\) repeated reconciliations per instance, 
averaged into a single reported value. 
Each instance runs its reconciliation task in an isolated fresh process, to reduce contamination of memory usage.
Memory (RSS) is sampled through
\texttt{/proc/self/status} immediately before and after each bench run, so
we measure the memory increment due to RBSR running and accessing the storage.
Protocol parameters are held fixed across backends; 
experiments use Negentropy's default split fanout $b=16$, 
and the default explicit-enumeration cutoff threshold $t=32$.
All computed outputs were verified against a ground truth.
Moreover, we verified that all RBSR reconciliations on the same instance produce exactly the same byte sequence. 
Thus, the comparison isolates the local cost of realizing
the reconciliation oracle, rather than measuring different
communication behavior.

\noindent\textbf{Scope of the evaluation.}
The evaluation is intentionally scoped to reconciliation-heavy,
single-machine workloads under a fixed protocol configuration.
Its purpose is to isolate the local cost of realizing the storage, 
not to characterize concurrent transactions, crash recovery,
cold-cache behavior, or cross-engine portability.
Those dimensions are important, but they answer a different systems
question and are therefore left as future work.

\subsection{Scenario families}

Let
\[
X_{\mathrm{in}} := X \cap \range{l}{u}, \qquad
Y_{\mathrm{in}} := Y \cap \range{l}{u},
\]
and
\[
X_{\mathrm{out}} := X \setminus (X \cap \range{l}{u}), \qquad
Y_{\mathrm{out}} := Y \setminus (Y \cap \range{l}{u}).
\]
The six scenario families are summarized in
Table~\ref{tab:scenario-families}. 
Each scenario is parametric and has \(i=1,\ldots,8\) instances.
The progression is from linearly scaled dense and sparse
reference families, to larger slice-growth families, to disagreement-heavy
stress families. In \texttt{stress\_dyn\(_i\)}, both the intersection set and the symmetric difference inside the slice grow quadratically in \(i\).

\subsection{Results}

Table~\ref{tab:absolute-results} reports family-wise arithmetic means for
the absolute metrics, while Table~\ref{tab:relative-results} reports
family-wise geometric means of per-scenario ratios relative to
\texttt{AELMDB}. These two views are complementary: the absolute table
shows scale, while the relative table makes the cross-backend structure
easier to read.

The main systems result is that \texttt{AELMDB} is the strongest persistent
backend for reconciliation on all evaluated scenario families. Relative to
\texttt{BTreeLMDB}, its reconciliation-time advantage grows from
\(4.69\times\) on \texttt{base\_dense\(_i\)} to \(13.98\times\) on
\texttt{stress\_dyn\(_i\)}. The no-window subrange ablation
\texttt{NoWndAELMDB} remains consistently slower than \texttt{AELMDB}, with
relative reconciliation factors between \(1.08\times\) and \(1.27\times\),
confirming that the window-subrange helper path contributes materially in
addition to the aggregate-enabled storage representation itself.

In terms of memory, \texttt{AELMDB} exhibits a lower
\(S_{\mathrm{RSS}}\) than the other persistent backends in all six evaluated
scenario families. This suggests that the aggregate-augmented
B\textsuperscript{+}-tree also reduces memory pressure under the tested
workloads w.r.t. the baseline \texttt{BTreeLMDB}.
Measured as a family-wise geometric mean of per-scenario ratios, the
\texttt{BTreeLMDB}/\texttt{AELMDB} $S_{\mathrm{RSS}}$ factor rises monotonically from
\(1.06\times\) on \texttt{base\_dense\(_i\)} to \(1.36\times\) on
\texttt{stress\_dyn\(_i\)}. 

Preparation (i.e. insertions in the database) still slightly favors
\texttt{BTreeLMDB} over \texttt{AELMDB}, with relative factors between
\(0.89\times\) and \(0.96\times\). This is not surprising, as
the B\textsuperscript{+}-tree of the \texttt{AELMDB} backend requires
delta updates of the aggregates over the entire root-to-leaf path at 
every database insertion.
However, the increased cost remains limited.
The used disk size \(S_{\mathrm{disk}}\) remains nearly tied between the tested persistent backends, with negligible differences. 
The conclusion is therefore that, on the evaluated
reconciliation-heavy workloads, AELMDB substantially improves the local cost
of range reconciliation while remaining space-competitive and incurring only
a modest increase in preparation cost.

\begin{table}[t]
\caption{Absolute benchmark results.
Bold indicates the best value among \emph{persistent} backends (i.e. \texttt{Vector} is excluded).
Values are arithmetic means over the \(8\) instances
in each family.
}
\label{tab:absolute-results}
\centering
\small
\begin{tabular}{|l|l|r|r|r|r|}
\hline
Family &
Backend &
\multicolumn{1}{c|}{\stackanchor{$T_{\mathrm{prep}}$}{(ms)}} &
\multicolumn{1}{c|}{\stackanchor{$T_{\mathrm{rec}}$}{(ms)}} &
\multicolumn{1}{c|}{\stackanchor{$S_{\mathrm{disk}}$}{(MiB)}} &
\multicolumn{1}{c|}{\stackanchor{$S_{\mathrm{RSS}}$}{(MiB)}} \\
\hline\texttt{base\_dense\(_i\)}
 & \texttt{Vector} & 1.663 & 0.109 & 0.218 & 7.654 \\
 & \texttt{BTreeLMDB} & \bf 15.928 & 0.930 & 0.452 & 8.199 \\
 & \texttt{NoWndAELMDB} & 17.049 & 0.205 & \bf 0.445 & 7.782 \\
 & \texttt{AELMDB} & 17.049 & \bf 0.188 & \bf 0.445 &\bf 7.722 \\
\hline\texttt{base\_sparse\(_i\)}
 & \texttt{Vector} & 3.086 & 0.126 & 0.435 & 8.086 \\
 & \texttt{BTreeLMDB} & \bf 23.275 & 1.288 & 0.877 & 8.577 \\
 & \texttt{NoWndAELMDB} & 25.908 & 0.286 & \bf 0.864 & 7.857 \\
 & \texttt{AELMDB} & 25.908 & \bf 0.263 & \bf 0.864 & \bf 7.747 \\
\hline\texttt{scale\_dense\(_i\)}
 & \texttt{Vector} & 7.068 & 0.186 & 1.074 & 9.367 \\
 & \texttt{BTreeLMDB} & \bf 38.522 & 3.002 & 1.991 & 9.979 \\
 & \texttt{NoWndAELMDB} & 42.782 & 0.404 & \bf 1.970 & 8.529 \\
 & \texttt{AELMDB} & 42.782 & \bf 0.342 & \bf 1.970 & \bf 8.421 \\
\hline\texttt{scale\_sparse\(_i\)}
 & \texttt{Vector} & 7.189 & 0.161 & 1.327 & 9.864 \\
 & \texttt{BTreeLMDB} & \bf 47.863 & 1.882 & 2.598 & 10.173 \\
 & \texttt{NoWndAELMDB} & 50.521 & 0.402 & \bf 2.566 & 8.144 \\
 & \texttt{AELMDB} & 50.521 & \bf 0.314 & \bf 2.566 & \bf 8.041 \\
\hline\texttt{stress\(_i\)} 
 & \texttt{Vector} & 22.787 & 0.421 & 2.655 & 12.585 \\
 & \texttt{BTreeLMDB} & \bf 79.491 & 7.588 & 4.609 & 13.196 \\
 & \texttt{NoWndAELMDB} & 85.675 & 0.869 & \bf 4.577 & 9.895 \\
 & \texttt{AELMDB} & 85.675 & \bf 0.710 & \bf 4.577 & \bf 9.776 \\
\hline\texttt{stress\_dyn\(_i\)}
 & \texttt{Vector} & 78.247 & 19.708 & 9.650 & 33.241 \\
 & \texttt{BTreeLMDB} & \bf 301.308 & 1197.161 & 15.779 & 29.018 \\
 & \texttt{NoWndAELMDB} & 319.332 & 74.800 & \bf 15.706 & 20.214 \\
 & \texttt{AELMDB} & 319.332 & \bf 62.927 & \bf 15.706 & \bf 20.033 \\
\hline
\end{tabular}
\end{table}

\begin{table}[t]
\caption{Relative results versus \texttt{AELMDB}. 
Each entry is the geometric mean, over the \(8\)
scenarios of the family, of the corresponding per-scenario ratio. 
Values larger than \(1\) indicate that the numerator backend is slower or larger
than \texttt{AELMDB}; values smaller than \(1\) indicate the opposite.}
\label{tab:relative-results}
\centering
\small
\begin{tabular}{|l|r|r|r|r|r|r|}
\hline
Family & 
\multicolumn{1}{c|}{\stackanchor{$\frac{\texttt{BTree}}{\texttt{AELMDB}}$}{$T_{\mathrm{rec}}$}} &
\multicolumn{1}{c|}{\stackanchor{$\frac{\texttt{NoWnd}}{\texttt{AELMDB}}$}{$T_{\mathrm{rec}}$}} &
\multicolumn{1}{c|}{\stackanchor{$\frac{\texttt{Vector}}{\texttt{AELMDB}}$}{$T_{\mathrm{rec}}$}} &
\multicolumn{1}{c|}{\stackanchor{$\frac{\texttt{BTree}}{\texttt{AELMDB}}$}{$T_{\mathrm{prep}}$}} &
\multicolumn{1}{c|}{\stackanchor{$\frac{\texttt{BTree}}{\texttt{AELMDB}}$}{$S_{\mathrm{disk}}$}} &
\multicolumn{1}{c|}{\stackanchor{$\frac{\texttt{BTree}}{\texttt{AELMDB}}$}{$S_{\mathrm{RSS}}$}} \\
\hline
\texttt{base\_dense\(_i\)} & 4.69x & 1.08x & 0.59x & 0.94x & 1.02x & 1.06x \\
\texttt{base\_sparse\(_i\)} & 4.82x & 1.10x & 0.50x & 0.89x & 1.01x & 1.11x \\
\texttt{scale\_dense\(_i\)} & 7.27x & 1.17x & 0.55x & 0.90x & 1.01x & 1.18x \\
\texttt{scale\_sparse\(_i\)} & 5.80x & 1.27x & 0.51x & 0.96x & 1.01x & 1.25x \\
\texttt{stress\(_i\)} & 9.20x & 1.21x & 0.58x & 0.92x & 1.01x & 1.32x \\
\texttt{stress\_dyn\(_i\)} & 13.98x & 1.17x & 0.39x & 0.93x & 1.01x & 1.36x \\
\hline
\end{tabular}
\end{table}

\section{Discussion and Open Directions}
\label{sec:discussion}

The main point of this paper is to clarify the \emph{local} storage
requirements behind efficient RBSR
\cite{meyer2023rbsr}. The improvement studied here does
not come from changing the communication logic of RBSR, but from changing
how its reconciliation oracle is realized. Once the required operations
are stated explicitly (composable range summaries, efficient navigation
by ordered position, and small-range enumeration in the recursive
refinement process \cite{meyer2023rbsr}), it becomes
natural to view RBSR through the abstraction of a
range-summarizable order-statistics store.
The RSOS abstraction is largely independent of the specific form of the
protocol-level comparison rule: the storage layer maintains composable
aggregates and order-statistics support, while the concrete protocol decides
how those aggregates are turned into range-comparison values.

Aggregate-augmented
B\textsuperscript{+}-trees are one coherent realization of this view, and
AELMDB shows that such a realization can be made practical inside a
persistent memory-mapped engine by extending LMDB with optional in-tree
aggregate metadata \cite{aelmdb,lmdb}. However, the storage-theoretic
argument is not specific to LMDB, and in practice the gains will depend
on workload characteristics, especially the balance between frequent
reconciliation queries and the maintenance cost of aggregate metadata
under updates.

\texttt{AELMDB} is a reusable storage package rather than a
one-off synthetic data structure. 
The observed speedups are evidence that the
RSOS view can be realized effectively in a persistent prototype 
built on top of LMDB.
Within the evaluated workloads, the gains appear on the
hot reconciliation path that motivated the design, while preparation cost
remains close and storage overhead remains essentially tied.

Several open directions follow. On the analytical side, it would be
useful to derive instance-sensitive bounds for local work that depend not
only on $|\Diff(X,Y)|$ but also on the ordered shape of the mismatches
and on the refinement tree induced by the protocol
\cite{meyer2023rbsr}. On the systems side, it remains
important to evaluate the RSOS view across other storage engines and
update regimes, and to understand more systematically how split
policies, branching factors, and enumeration thresholds interact with the
underlying index; Willow's three-dimensional adaptation\,\cite{willow3drbsr} 
makes explicit,
for example, that efficient splits should balance item counts rather than
merely cut a geometric domain in half. More
conceptually, an interesting question is how far the RSOS abstraction
extends beyond homomorphic summary constructions, especially in light of
recent work showing that RBSR can also be realized with conventional hash
functions rather than homomorphic ones
\cite{meyer2024nonhomomorphic}. 
Finally, extensions to richer ordered
domains (such as composite-key or multidimensional
reconciliation) appear promising, but would require a clearer theory of
balancing and summarization beyond the one-dimensional setting considered
here \cite{willow3drbsr}.

\section{Conclusion}
\label{sec:conclusion}

This paper argued that the natural backend abstraction for
range-based set reconciliation is a range-summarizable order-statistics
store (RSOS), which combines composable range summaries with navigation by
rank. Framed in this way, the storage requirements of efficient recursive
reconciliation can be stated explicitly as an abstract interface rather
than left implicit in a particular implementation strategy.

It then showed that aggregate-augmented
B\textsuperscript{+}-trees realize this abstraction naturally: subtree
counts provide the order-statistics structure needed for balanced
partitioning, while composable aggregates provide efficient range-summary
support. This yields a direct explanation for why such trees are a good
match for RBSR and supports the derived bounds on local reconciliation
work.

The AELMDB implementation and its use through the Negentropy protocol
provide concrete systems evidence for this design. Across the evaluated
workloads, the aggregate-aware design
confirms the practical benefit of moving range-summary and rank-navigation
support into the storage engine, compared with a persistent baseline based on
auxiliary tree structures. A more general analysis about concurrency, crash recovery,
cold-cache behavior, and cross-engine portability remains for future work.

More broadly, the paper helps connect the protocol structure of
range-based reconciliation with the design of persistent ordered storage,
showing that efficient RBSR requires realizing the full storage interface
needed by recursive range refinements: composable range aggregation,
order-statistics navigation, and exact enumeration on small residual ranges.

\bigskip
\subsection*{Code Availability}
\begin{itemize}
    \item AELMDB: \url{https://github.com/amparore/aelmdb}
    \item Negentropy with AELMDB: \url{https://github.com/amparore/negentropy-aelmdb}
    \item C++ wrapper API: \url{https://github.com/amparore/lmdbxx-aelmdb}
    \item Benchmark code: \url{https://github.com/amparore/bench-aelmdb}
\end{itemize}

\bibliographystyle{plain}
\bibliography{bibtex}{}

\clearpage

\appendix
\section*{Appendix}

\section{Range aggregate example}
\label{sec:example:rangeaggr}

\begin{example}
As a concrete instantiation, let
\[
\U = \mathbb{N} \times \mathbb{Z}_{256},
\]
where \(\mathbb{Z}_{256}\) is a set of 8-bit hexadecimal identifiers, and let
\(\prec\) be the lexicographic order on \(\U\). 
Consider the replica state
\[
X =
\{
(10,\texttt{a1}),
(10,\texttt{f3}),
(11,\texttt{1c}),
(13,\texttt{7b})
\}.
\]

Let the element-summary monoid be
\[
(M,\oplus,0_M) = (\mathbb{Z}_{256},\; + \bmod 256,\; 0),
\]
and define
\[
\phi(t,\texttt{id}) := \texttt{id},
\]
where each identifier is interpreted as an element of \(\mathbb{Z}_{256}\).
Thus the order on records is provided by \(\prec\), while identifiers
determine the summary contribution.

Now consider the half-open range
\[
\range{10}{13},
\]
where the chosen bounds select the records of \(X\) whose first component
lies in \(\range{10}{13}\). Then
\[
X \cap \range{10}{13}
=
\{
(10,\texttt{a1}),
(10,\texttt{f3}),
(11,\texttt{1c})
\}.
\]
Its summary is
\[
A\bigl(X \cap \range{10}{13}\bigr)
=
\left(
3,\;
\texttt{a1} + \texttt{f3} + \texttt{1c} \bmod 256
\right).
\]
Interpreting the identifiers as hexadecimal bytes,
\[
\texttt{a1} + \texttt{f3} + \texttt{1c}
= 161 + 243 + 28 = 432 \equiv 176 \pmod{256},
\]
so the aggregate is
\[
A\bigl(X \cap \range{10}{13}\bigr) = (3,\texttt{b0}).
\]



This example illustrates the storage-side aggregate object maintained by the
RSOS abstraction; concrete protocols may derive their compared range values
from this aggregate in different ways.
\end{example}

\section{Bounds for reconciliation over RSOS B\textsuperscript{+}-trees}
\label{sec:bounds}

This section derives local cost bounds for RBSR when the underlying oracle is
realized by an aggregate-augmented B\textsuperscript{+}-tree, as in
Section~\ref{sec:bptree}. The purpose is to connect the protocol-level
recursion of Section~\ref{sec:protocol} with the storage-level
$O(h)$ realization cost of the RSOS operations.
We first bound the cost of one reconciliation step, then sum this cost over one execution tree to obtain an execution-sensitive local bound.

\subsection{Setting and assumptions}

Fix two finite replica states $X,Y \subseteq \U$, and let
\[
R_0 := \range{l_0}{u_0}
\]
be the outer range on which reconciliation is performed. Write
\[
X_0 := X \cap R_0,
\qquad
Y_0 := Y \cap R_0,
\qquad
n_X := |X_0|,
\qquad
n_Y := |Y_0|,
\qquad
n := \max(n_X,n_Y).
\]
Let $b \ge 2$ be the branching parameter of the RBSR refinement and
$t \ge 1$ the explicit-enumeration threshold at which recursion stops 
and the remaining items are listed directly.

Assume that both peers are represented by RSOS oracles realized by
aggregate-augmented B\textsuperscript{+}-trees of branching factor $B$ and
height
\[
h = \Theta(\log_B n),
\]
under the standard bounded-page-size assumption\footnote{Here we assume the usual external-memory setting in which page size
is fixed independently of the dataset size $n$. The branching factor $B$ is
therefore bounded as well, so an $O(Bh)$ term reduces to $O(h)$.}. 
By Theorem~\ref{thm:bptree-rsos}, each oracle supports:
\begin{itemize}
    \item $\Aggregate(l,u)$, $\Rank(z)$, and $\Select(r)$ in $O(h)$ time;
    \item $\Enumerate(l,u)$ in $O(h+k)$ time, where
    $k = |X \cap \range{l}{u}|$ (or the corresponding quantity on $Y$).
\end{itemize}

Throughout this section, we analyze a \emph{static} reconciliation run:
the sets $X$ and $Y$ do not change during the run.

\subsection{The reconciliation tree}

The right combinatorial object for the analysis is the recursion tree
generated by the protocol.

\begin{definition}[Reconciliation tree]
\label{def:reconciliation-tree}
Fix one execution of RBSR on the outer range $R_0$. Its
\emph{reconciliation tree} $\mathcal{T}$ is the rooted tree defined as
follows.
\begin{itemize}
    \item The root is $R_0$.
    \item Each queried range is a node of $\mathcal{T}$.
    \item A node is a leaf if the protocol returns either
    \textsc{Skip} or \textsc{IdList} on that range.
    \item A node is internal if the protocol returns \textsc{Split}; its
    children are the child ranges returned by the split.
\end{itemize}
\end{definition}

Thus the reconciliation tree records exactly the recursive refinement
performed by the protocol. Its leaves are the ranges that are resolved
without further subdivision.

Let:
\begin{align*}
Q &:= \text{total number of queried ranges} = |\mathcal{T}|,\\
I &:= \text{number of internal (\textsc{Split}) nodes},\\
L &:= \text{number of leaves} = Q-I.
\end{align*}

Since each internal node has at most $b$ children, the basic tree relation
gives
\begin{equation}
\label{eq:tree-basic}
L \le 1 + (b-1)I,
\qquad
Q = I + L \le 1 + bI.
\end{equation}

This simple inequality will let us convert bounds on the number of split
nodes into bounds on the total number of queried ranges.

\subsection{Cost of one protocol step}

We first bound the local work performed by one peer when answering one range
query.

\begin{lemma}[Cost of one responder step]
\label{lem:single-step-cost}
Consider one call to $\Call{Respond}{O,l,u,f,b,t}$ in
Algorithm~\ref{alg:rbsr-step}, where $O$ is realized by an
aggregate-augmented B\textsuperscript{+}-tree of height $h$.
\begin{enumerate}
    \item A \textsc{Skip} answer costs $O(h)$.
    \item An \textsc{IdList} answer costs $O(h+k)$, where
    $k = |S \cap \range{l}{u}|$ for the responder's local set $S$.
    \item A \textsc{Split} answer costs $O(bh)$.
\end{enumerate}
\end{lemma}

\begin{proof}
In all cases, the responder begins by evaluating
$\Aggregate(l,u)$ and its comparison value, which costs $O(h)$.

If the comparison values match, the answer is \textsc{Skip}, so the total cost is
$O(h)$.

If the range mismatches but the responder's local cardinality is at most
$t$, the responder returns \textsc{IdList} together with
$\Enumerate(l,u)$. By Theorem~\ref{thm:bptree-rsos}, this costs
$O(h+k)$, where $k$ is the number of enumerated items.

If the responder returns \textsc{Split}, the work consists of:
\begin{itemize}
    \item one parent $\Aggregate(l,u)$ call;
    \item one $\Rank(l)$ call;
    \item $b-1$ calls to $\Select$ for the interior cut points;
    \item up to $b$ child calls to $\Aggregate$ in order to obtain 
        the comparison values of the returned child ranges.
\end{itemize}
Each of these operations costs $O(h)$, so the total is
$O((2b+1)h)=O(bh)$.
\end{proof}

The importance of Lemma~\ref{lem:single-step-cost} is that, once the RSOS is
realized by an augmented B\textsuperscript{+}-tree, every non-output
protocol step costs only logarithmic local work. The remaining question is
therefore how many such steps one reconciliation run can generate.

\subsection{Execution-sensitive local cost}

We now sum the single-step cost over the reconciliation tree.

For one execution, let:
\begin{align*}
L_{\mathrm{skip}} &:= \text{number of \textsc{Skip} leaves},\\
L_{\mathrm{id}} &:= \text{number of \textsc{IdList} leaves},\\
K &:= \text{total number of explicitly enumerated items returned by this peer
over all \textsc{IdList} leaves}.
\end{align*}
Thus
\[
L = L_{\mathrm{skip}} + L_{\mathrm{id}}.
\]

\begin{theorem}[Execution-sensitive responder-side local cost]
\label{thm:execution-sensitive-cost}
Consider one static reconciliation run over an RSOS realized by an
aggregate-augmented B\textsuperscript{+}-tree of height $h$.
Then the total responder-side local work performed by one peer 
while answering queries is
\begin{equation}\label{eq:execution-sensitive-main}
    T_{\mathrm{loc}} = O\!\left( hL + bhI + K \right)
\end{equation}
and since $L \leq Q$, also
\begin{equation}\label{eq:execution-sensitive-simplified}
    T_{\mathrm{loc}} = O\!\left( hQ + bhI + K \right)
\end{equation}
If $b$ is treated as a fixed protocol constant, this simplifies to
\[
T_{\mathrm{loc}} = O(Qh + K).
\]
\end{theorem}

\begin{proof}
Each \textsc{Skip} leaf contributes $O(h)$ by
Lemma~\ref{lem:single-step-cost}.
Each \textsc{IdList} leaf contributes $O(h+k_v)$, where $k_v$ is the number
of items explicitly returned at that leaf. Summing over all
\textsc{IdList} leaves gives $O(hL_{\mathrm{id}} + K)$.
Each internal \textsc{Split} node contributes $O(bh)$.

Summing these contributions over the reconciliation tree gives
\[
T_{\mathrm{loc}}
=
O\!\left(
hL_{\mathrm{skip}}
+
hL_{\mathrm{id}}
+
bhI
+
K
\right)
=
O\!\left(
hL + bhI + K
\right).
\]
Since $Q = I+L$, this is
\[
T_{\mathrm{loc}} = O(hQ + bhI + K).
\]
Finally, $I \le Q$, so
\[
T_{\mathrm{loc}} = O(bQh + K).
\]
\end{proof}

\begin{remark}[Receiver-side work at \textsc{IdList} leaves]
Theorem~\ref{thm:execution-sensitive-cost} counts the work of producing replies.
If one instead wishes to bound the full bilateral local work of a reconciliation
run, one must add the receiver-side cost of comparing each returned list with
the receiver's local ordered contents on the same range.
\end{remark}

Equation~\eqref{eq:execution-sensitive-simplified} is the basic local cost
bound of the section. It cleanly separates:
\begin{itemize}
    \item the \emph{protocol-side complexity}, captured by the number $Q$ of
    queried ranges and the total explicit output size $K$;
    \item the \emph{storage-side complexity}, captured by the factor $h$,
    i.e.\ the logarithmic height of the RSOS B\textsuperscript{+}-tree.
\end{itemize}

\paragraph{Bounded-threshold simplification.}
Since every \textsc{IdList} leaf returns at most $t$ items, one always has
\[
K \le t L_{\mathrm{id}} \le tQ.
\]
Hence, for fixed $b$ and fixed threshold $t$,
\begin{equation}
\label{eq:qh-bound}
T_{\mathrm{loc}} = O(Qh).
\end{equation}

This is the cleanest high-level statement: once the protocol recursion tree
has size $Q$, an RSOS B\textsuperscript{+}-tree executes the run in
logarithmic local time per queried range.

\end{document}